\documentclass[12pt]{amsart}
\usepackage{amsmath,amssymb, amscd, graphicx,youngtab}
\usepackage{enumerate}

\makeatletter
\@addtoreset{equation}{section}

\usepackage{youngtab}

\def\dd{\mathrm d}

\def\wt{\mathrm {wt}}

\def\Spec{\mathrm {Spec\ }}

\def\CC{{\mathbb C}}

\def\NN{{\mathbb N}}
\def\RR{{\mathbb R}}
\def\ZZ{{\mathbb Z}}
\def\QQ{{\mathbb Q}}

\def\WWf{{{\mathcal W}^{}}}

\def\fZ#1{{f^{(Z)}_{#1}}}

\def\wpg#1{{{\wp}^{}_{#1}}}

\def\SS{{\mathcal S}}
\def\WW{{\mathcal W}}

\def\cJ{\mathcal{J}}

\def\cO{\mathcal{O}}
\def\cL{\mathcal{L}}

\def\tP{\tilde{P}}

\def\Ng{{N}^{(g)}}
\def\NHf{{N_{H^1}^{}}}

\def\omegagp#1{{\omega^{{\prime}}_{#1}}}
\def\omegagpp#1{{\omega^{{\prime\prime}}_{#1}}}
\def\omegafp#1{{\omega^{{\prime}}_{#1}}}
\def\omegafpp#1{{\omega^{{\prime\prime}}_{#1}}}

\def\etagp#1{{\eta^{{\prime}}_{#1}}}
\def\etagpp#1{{\eta^{{\prime\prime}}_{#1}}}

\def\phiHf#1{{{\phi_{H^1}^{}}_{#1}}}

\def\nuIg#1{{\nu^{{I}}_{#1}}}

\def\nuIf#1{{\nu^{{I}}_{#1}}}

\def\nuIIg#1{{\nu^{{II}}_{#1}}}

\def\nuIIf#1{{\nu^{{II}}_{#1}}}

\def\muc{{b}}

\def\lone#1{{\lambda^{(1)}_{#1}}}
\def\ltwo#1{{\lambda^{(2)}_{#1}}}

\def\uab{{\hat{u}}}
\def\hu{{\hat{u}}}

\def\mug#1{{\mu_{#1}}}

\def\Sigmaf{{\Sigma^{{}}}}
\def\Omegaf{{\Omega^{{}}}}

\def\Omegag{{\Omega^{{}}}}
\def\Pif{{\Pi^{{}}}}
\def\Ff{{F^{{}}}}

\def\Fg{{F^{{}}}}
\def\taug#1{{{\tau}^{}_{#1}}}

\def\alphag#1{{\alpha^{{}}_{#1}}}
\def\betag#1{{\beta^{{}}_{#1}}}
\def\gammag#1{{\gamma^{{}}_{#1}}}

\def\hzeta{{\hat{\zeta}}}

\def\LA{\langle}
\def\RA{\rangle}

\def\fGm{{\mathbb{G}_m}}

\newtheorem{theorem}{Theorem}[section]
\newtheorem{definition}[theorem]{Definition}

\newtheorem{proposition}[theorem]{Proposition}
\newtheorem{corollary}[theorem]{Corollary}
\newtheorem{remark}[theorem]{Remark}
\newtheorem{lemma}[theorem]{Lemma}

\newtheorem{example}[theorem]{Example}

\def\book#1{\rm{#1}, }
\def\paper#1{\textit{#1}, }
\def\jour#1{\rm{#1}, }
\def\yr#1{({\rm{#1}) }}
\def\vol#1{\textbf{#1}}
\def\pages#1{\rm{#1}}

\def\by#1{{\rm{#1}, }}

\begin{document}

\title{}
\begin{center}{\Large{
Sigma functions for a space curve of type $(3,4,5)$}}
\end{center}

\author{Shigeki Matsutani and Jiryo Komeda}


\date{ }

\maketitle

\begin{abstract}
In this article, a generalized Kleinian sigma function
for an affine $(3,4,5)$ space curve of genus 2
was constructed as the simplest example of the sigma function for
an affine space curve,
and in terms of the sigma function,
the Jacobi inversion formulae for the curve
are obtained.
An interesting relation between a space curve with a 
semigroup generated by
$(6,13,14,15,16)$ and Norton number associated with
Monster group is also mentioned with an Appendix by Komeda.
\end{abstract}


\bigskip
sigma function, space curve, Jacobi inversion formula

\bigskip
{\bf{MSC 2010:
14H05, 
14H42, 
14H50, 
14H55, 
20M07 
}}
\bigskip

\section{Introduction}
\bigskip

Recently 
 the Kleinian sigma function 
for hyperelliptic curves,
a natural generalization of the 
Weierstrass sigma function, is re-evaluated because
in terms of the sigma functions,
it is more convenient 
to investigate the properties of the abelian functions 
and their interesting properties are revealed naturally
 \cite{BEL1,Ma01,EEMOP}.

Further in \cite{EEL}, Enolskii, Eilbeck, and Leykin discovered a construction
which generalizes the Kleinian sigma 
function associated with hyperelliptic curves to one for 
an affine $(r,s)$ plane curve, where
$r$ and $s$ $(r < s)$ are coprime positive integers 
$g=(r-1)(s-1)/2$.
In \cite{EEL},
they, firstly, constructed the fundamental differential of the second kind 
over an affine $(r,s)$ plane curve and 
using it, obtained the Legendre relation as the symplectic structure
over the curve.
Using the Legendre relation, they defined the generalized
Kleinian sigma function over  the image of the abelian map $\CC^g$.
They also found the natural Jacobi inversion formulae in terms of their
sigma function.
We call the construction {\it{EEL construction}} in this article.
Using the EEL construction,
we have some interesting results \cite{MP1,MP2}.

In this article, we consider 
a generalized Kleinian sigma function
for an affine $(3,4,5)$ space curve of genus 2, 
which is the simplest affine space curve.
Our purpose of this article is to show that the
sigma function is also defined 
for an affine space curve as we can do for plane curves.

Following the EEL-construction,
we define the fundamental differential of the second kind
over it and obtain the Legendre relation as the 
symplectic structure over it.
With the abelian map to $\CC^2$, we show that
the symplectic structure determines the sigma function.
Further using the sigma function,
we obtain the Jacobi inversion formulae for the curve and the
Jacobian following the previous work \cite{MP1,MP2}.

It means that the generalization of the sigma functions for 
the affine plane curves to ones for the space curves is basically
possible and is useful.
Recently, Korotkin with Shramchenko defined a sigma function
for a compact Riemann surface \cite{KS} but it is not directly associated with
an algebraic curve. Further Ayano introduced sigma functions 
for space curves of special class \cite{A}, which are called telescopic curves,
but the class does not include this (3,4,5) curve.

In Remark \ref{rmkF}, we also show a problem of a space curve 
associated with the semigroup generated by
$(6,13,14,15,16)$ with an Appendix by Komeda.
The semigroup might be related to Norton number associated with
the Monster group, the simple largest sporadic finite group \cite{MS}.

%
%

\section{Preliminary}

\subsection{Numerical semigroup}
Here we give a short overview of recent study of the
numerical semigroups as sub-semigroups of non-negative integers $\NN_0$
related to algebraic curves.
We call an additive semigroup in $\NN_0$
{\it{ numerical semigroup}}  if its complement in $\NN_0$ is finite. 
For a numerical semigroup $H = H(M)$ generated by $M$,
the number of elements of $L(H):=\NN_0\setminus H$ is called genus
and $L(H)$ is called gap sequence.
For example, we have semigroups $H_{2}$, $H_{4}$, $H_{12}$ generated by 
$ M_{2}:=\LA3, 4, 5\RA$,
$M_{4}:=\LA3, 7, 8\RA$,
$M_{12}:=\LA6, 13, 14, 15,  16\RA$
respectively whose genera are $g(H_{g})$ for $g = 2, 4, 12$ due to
$
L(H_{2})=\{1, 2\},$
$L(H_{4})=\{1, 2, 4, 5\},$
$
L(H_{12})=\{1, 2, 3, 4, 5, 7, 8, 9, 10, 11, 17, 23 \} .
$

For a complete non-singular irreducible curve $C$ of genus $g$
over an algebraically
closed field $k$ of characteristic 0, the field of its rational 
functions $k(C)$, and a point $P \in C$, we define 
\begin{equation}
 H(P):= \{n \in \NN_0\ |\ \mbox{there exists } f \in k(C)
                          \mbox{ such that } (f)_\infty = n P\ \}
\label{eq:H(P)}
\end{equation}
which is called the Weierstrass semigroup of the point $P$.
If $L(H(P)):=\NN_0\setminus H(P)$ differs from the set 
$\{1, 2, \cdots, g\}$, 
we call $P$ Weierstrass point of $C$.

A numerical semigroup $H$
is said to be Weierstrass 
if there exists a pointed algebraic curve $(C, P)$ such that 
$H=H(P)$.
Hurwitz considered whether every numerical semigroup H is
Weierstrass. This was a long-standing problem but
Buchweitz finally showed that every $H$ is not Weierstrass. 
His first counterexample is the semigroup $H_B$ generated by 
13, 14, 15, 16, 17, 18, 20, 22 and 23, whose genus is 16. 
Thus in general, it is not so trivial whether a given semigroup is
Weierstrass or not. Komeda has been investigated this problem
with Ohbuchi and Kim \cite{KK,K83,K99,KO04}.

\subsection{Commutative Algebra}\label{sec:2.2}
Here
we review a normal ring and normalization in commutative ring \cite{Mat}.
We assume that every ring is a commutative ring with unit.

$B$ is a ring and $A$ is a subring of $B$.
$B$ is said to be {\it{extension}} of $A$.
An element $b$ of $B$ is said to be {\it{integral}} over $A$ if
$b$ satisfies a monic polynomial over $A$, i.e.,
there exist $n$ and $\{a_i\}_{i=1,\dots,n}\in A$ such that
$
 b^n + a_1 b^{n-1} + \cdots a_n =0.
$

We say that $B$ is {\it{integral}} over $A$, or 
$B$ is an {\it{integral ring}} over $A$, or 
$B$ is an {\it{integral extension}} of $A$ 
if every element $b$ of $B$ is integral over $A$.

An {\it{integral closure}} in $B$ over $A$ is defined by
$
\tilde A:=\{b \in B\ |\ b \mbox{ is integral over } A\}.
$
If $A=\tilde A$, $A$ is {\it{integral closed}} in $B$.

\begin{definition}
$A$ is a ring and $Q(A)$ is a quotient ring of $A$.
We assume that $A$ is an integral domain.
$A$ is {\rm{normal}} if $A$ is integral closed in $Q(A)$, i.e.,
for
$
\tilde A:=\{q \in Q(A)\ | \ $ there exist $n$ and $a_i \in A$ such that 
$
 q^n + a_1 q^{n-1} + \cdots a_n =0
$ \}, 
$A=\tilde A$.

We define the minimum extension $\hat A$ of $A$ in $Q(A)$ so
that $\hat A$ is integral closed in $Q(A)$. 
We say that $\hat A$ is {\rm{normalization}} of $A$ or the
{\rm{normalized ring}} of $A$. 
\end{definition}

Through the correspondence between an algebraic variety and
a commutative ring,
we have the well-known normalization theorem \cite[p.5, p.68]{Gri}:
\begin{theorem}\label{thm:normal}
For any irreducible algebraic curve $X \subset P^2\CC$, there 
exists a compact Riemann surface $\tilde X$ and a holomorphic mapping
$s: \tilde X \to P^2\CC$ such that $s(\tilde X)=X$ and
$s$ is injective on the inverse image of the set of smooth
points of $X$.
Further the Riemann surface is unique up to its isomorphism;
if there are two Riemann surfaces $\tilde X$ and $\tilde X'$
given by normalizations of $X$, there is
a biholomorphic from $\tilde X$ to $\tilde X'$.
\end{theorem}

As examples of Theorem \ref{thm:normal}, we give three
examples. 

\begin{example}
{\bf{($x^3 - y^2$): }}
{\rm{
$R:=\CC[X,Y]/(X^3 - Y^2)$ is not normal because
$
  \frac{Y}{X} \in \tilde R \setminus R \subset Q(R)
$
due to $\left(\frac{Y}{X}\right)^2 - X = 0$.
Since
$R \approx \CC[t^2, t^3]$, the normalized ring is
$\hat R = \CC[t]$.
}}
\end{example}


\noindent
\begin{example}
{\bf{( $y^3 = x^5 -1$ and $w^3 = z - z^6$): }}
{\rm{
Following Theorem \ref{thm:normal}, we consider the
covering of a curve of $f(x,y) := y^3 - x^5 + 1$.
Let us consider a homogeneous polynomial
$F(X,Y,Z) := Y^3 Z^2 - X^5 + Z^5 \in \CC[X,Y,Z]$. 
Around $Z\neq 0$, we have
$
F(X,Y,Z) = Z^5\left( \frac{Y^3}{Z^3} - \frac{X^5}{Z^5} + 1 \right)
$
and thus by regarding that $x = X/Z$ and $y = Y/Z$, we have
$F(X,Y,Z) = Z^5 f(X/Z, Y/Z)$.
$R_0 := \CC[x, y]/(f(x,y))$ is a normal ring.
On the other hand,
around $Z = 0$ and $X \neq 0$, we have
$
F(X,Y,Z) = X^5\left( \frac{Y^3Z^2}{X^5} - 1 + \frac{Z^5}{X^5}\right),
$
and then we obtain a polynomial,
$
 g(w, z) = w^3 z^2 - 1 + z^5
$
by regarding $w = Y/X$ and $z = Z/X$. 
However $R_\infty := \CC[w, z]/(g(w,z))$ is not a normal ring.
As a vector space, $R_\infty$ is
$\CC 1 + \CC z + \CC z^2 +$ $  \cdots + \CC w + \CC w z + \CC w z^2 + \cdots$
$ + \CC w^2 + \CC w^2 z + \CC w^2 z^2 +  \cdots $ $ + \CC w^3 + \CC w^3 z$.
We  show that $q \in Q(R_\infty) \setminus R_\infty$ 
exists such that $q^n + a_1 q^{n-1} + \cdots a_n =0$ for
certain $a_i \in R_\infty$.
Noting
$\displaystyle{
 \frac{1}{1-z}g(w, z) = \frac{w^3 z^2}{1-z} + 1+z+z^2 + z^3 + z^4= 0 
\in Q(R_\infty),}
$
we consider
$
q := \frac{w^3}{1-z}+ \frac{1+z}{z^2} \in Q(R_\infty) \setminus R_\infty, 
$
which is integral over $R_\infty$.
By normalization, we define $\hat w := w z=y/x^2$.
$\hat R_\infty := \CC[\hat w, z]/(\hat g(w,z))$ is a normal ring,
where
$
 \hat g(\hat w, z) := \hat w^3 - z + z^6.
$
The minimal condition is obvious.
}}
\end{example}

\begin{example}\label{ex:3}
{\bf{ (a space curve; $y^3 = x^2(x^2-1)$
and $w^3 = x(x^2-1)^2$):}}
{\rm{
Let us consider a polynomial
$f(x, y) =y^3 - x^2(x^2-1)$ and show that
$R_0 := \CC[x, y]/(f(x,y))$ is not a normal ring.
As a vector space, $R_0$ is
$\CC 1 + \CC x + \CC x^2 +   \cdots $ $+ \CC y + \CC y x + \CC y x^2 + \cdots$
$ + \CC y^2 + \CC y^2 x + \CC y^2 x^2 + \cdots$.
We show that $w \in Q(R_0) \setminus R_0$ 
exists such that $w^n + a_1 w^{n-1} + \cdots a_n =0$
for certain $a_i$'s of $R_0$.
In other words, noting that $y \sim \sqrt[3]{x^2(x^2-1)}$
and $y^2 \sim x\sqrt[3]{x(x^2-1)}$,
one of $w$ is that  $ w:= \frac{y^2}{x}$
which is integral over $R_0$ because
$w^3 = \frac{y^6}{x^3} =  x(x^2-1)^2$ or $w^3 -  x(x^2-1)^2 = 0 \in R_0$.
Let $g(x,w) = x (x^2-1)^2$.
Noting the relations
that $w = \frac{y^2}{x}$, $y = \frac{w^2}{x^2-1}$, and $wy = (x^2-1)x$,
we have 
$
\hat R_0:=\CC[x,y,w]/(f_1(x,y,z), f_2(x,y,z)
 f_3(x,y,z)),
$
as the  normalized ring of $R_0$,
where
$f_1(x,y,w) = y^2 - w x$,
$f_2(x,y,w) = wy - (x^2-1)x$, 
and 
$f_3(x,y,w) = w^2 - y (x^2-1)$.
The minimal condition is also obvious.
This example corresponds to the special case of
the affine $(3,4,5)$ space curve in this article.
Due to Theorem \ref{thm:normal},
the corresponding Riemann surface uniquely exists up to an isomorphism.
}}
\end{example}

\section{A Curve (3,4,5)}

Since $H_2$ generated by $\LA3, 4, 5\RA$ is Weierstrass and 
is the simplest semigroup whose cardinality of the generators is
greater than 2, 
we consider a curve $C(H_2)$ explicitly in order to construct the
sigma functions for $C(H_2)$ following the
EEL construction.

Following Theorem \ref{thm:normal},
in order to construct a non-singular curve $X_2 = C(H_2)$,
we consider two singular curves $X_3$ and $X_4$ generated by
$\infty$ points and 
the zeroes of
$$
   f_{3,12}(x, y_4) := y_4^3 - k_4(x) ,
   \quad f_{4,15}(x, y_5) := y_5^3 - k_5(x) 
$$
where 
$k_4(x) := k_2(x) k_1(x)^2$, $k_5(x) := k_2(x)^2 k_1(x)$,
$k_2(x) := (x - \muc_1)(x - \muc_2)
= x^2 + \ltwo{1} x + \ltwo{2}$, and
$k_1(x) := (x - \muc_0) 
= x + \lone{1}$ 
for finite $\muc_a \in \CC$ $(a = 1, 2, 3)$
which is distinct from each other.
Let us consider commutative rings $R_3:=\CC[x, y_4]/(f_{3,12}(x,y_4))$ and
$R_4:=\CC[x, y_5]$ $/(f_{4,15}(x,y_5))$ related to $X_3$ and $X_4$ respectively.
These  genera of the semigroups associated with their Weierstrass non-gap
sequences at $\infty$-points are three and four respectively, though
the geometric genera are not.
Following the normalization in section 2, we normalize
$R_3$ and $R_4$.
Since in terms of the language of the commutative algebra \cite{Mat},
$\displaystyle{\frac{y_4^2}{(x-\muc_0)}}$
is integral over $R_3$ in $Q(R_3)$ and
$\displaystyle{\frac{y_5^2}{(x-\muc_0)(x-\muc_2)}}$
is integral over $R_4$ in $Q(R_4)$,
$R_3$ and $R_4$ are not normal rings.
Thus we will normalise them in $\CC[x, y_4, y_5]$ in the 
meaning of the commutative algebra \cite{Mat}
 (see Example \ref{ex:3} in \S \ref{sec:2.2}).

For the zeroes of $f_{3,12}(x,y_4)$ and
 $f_{3,15}(x,y_4)$,
we could have the relations,
\begin{equation}
	y_4 y_5 = k_2(x) k_1(x), \quad
           y_5 = \frac{y_4^2}{(x-\muc_0)}, \quad
           y_4 = \frac{y_5^2}{(x-\muc_1)(x-\muc_2)}\cdot \quad
\label{eq:rel345}
\end{equation}
Here for the primitive root
$\zeta_3$ $(\zeta_3^3 =1, \zeta_3 \neq 1)$, 
$\zeta_3$  acts on $X_3$ and $X_4$ respectively.
The first relation is chosen in the
possibilities $y_4 y_5 = \zeta_3^i k_2(x) k_1(x)$ $i=0,1,2$.

As a normalization of these singular curves,
we have the commutative ring,
$$
R_2 \equiv R:=\CC[x, y_4, y_5]/ (f_{8}, f_{9}, f_{10})
$$
and $X_2:=\Spec R$.
Here we define $f_{8}, f_{9}, f_{10} \in \CC[x, y_4, y_5]$ by
$$
f_{8} = y_4^2 - y_5 k_1(x), \quad
f_{9} = y_4 y_5 - k_2(x) k_1(x), \quad
f_{10} = y_5^2 - y_4 k_2(x)
$$
which are also regarded as the $2 \times 2$ minors of
$\displaystyle{
\left|\begin{matrix}
 k_2(x)  & y_{4} & y_{5} \\
 y_{4} & y_{5}  & k_3(x)\\
\end{matrix} \right|}$.
Here $\zeta_3$ acts on $X_2$ by 
$\hat \zeta_3 (x, y_4, y_5) = (x, \zeta_3 y_4, \zeta_3^2  y_5)$.

Let $X$ be the Riemann surface which is naturally obtained as
an extension of $X_2$ as mentioned in Theorem \ref{thm:normal}, i.e.,
$X = X_2 \cup \{\infty\}$ as a set.
It is noted that when $x$ diverges, $y_4$ and $y_5$ also diverge vise versa.
Thus the infinity point $\infty$ uniquely exists in $X$.
$\fGm$ acts on $R$ by setting $g_m^{-3} x$, $g_m^{-a} y_a$ for
$x$,  $y_a$,  $g_m \in \fGm$ and $a = 4, 5$.
By Nagata's Jacobi-method \cite{Mat}, it can be proved that
$X$ is non-singular.
 
Though they do not explicitly appear, we may also implicitly consider
parametrizations of $y_4$ and $y_5$ by
$
y_4 = w_2w_1^2,$ and  $y_5 = w_2^2 w_1,
$
where
$w_1^3 = k_1$ and $ w_2^3= k_2$.
When we consider $\tilde R:=\CC[x, w_1, w_2]/
(w_1^3 - k_1(x), w_2^3- k_2(x))$,
it is related to a natural covering of $X$.

\subsection{The Weierstrass gap and holomorphic one forms}
The Weierstrass gap sequences at $\infty$
are given in Table 1.
For the local parameter $t_\infty$ at $\infty$, we have
\begin{gather}
   x = \frac{1}{t_\infty^3}, \quad
   y_4 = \frac{1}{t_\infty^4}(1 + \dd_\ge(t_\infty)), \quad
   y_5 = \frac{1}{t_\infty^5}(1 + \dd_\ge(t_\infty)). \quad
\label{eq:x_t}
\end{gather}
Here for a given local parameter $t$ at some $P$ in $X$, 
 the series of $t$, whose orders of zero at $P$
are greater than $\ell$ or equal to $\ell$, is 
denoted by $\dd_\ge(t^\ell)$.
$H(\infty)$ in (\ref{eq:H(P)}) is $H(3,4,5)$
as Pinkham considered $(3,4,5)$ curve as the simplest example of the
numerical semigroup $H(3,4,5)$
\cite[Sec.14]{P}. Its monomial curve is defined by,
$
Z_4^2 = Z_3 Z_5,$ $ 
Z_4 Z_5 = Z_3^5,$ $
Z_5^2 = Z_3^3 Z_4,$
or the $2 \times 2$ minor of $
\left|\begin{matrix}
Z_3 & Z_4 & Z_5 \\ 
Z_4 & Z_5 & Z_3^2 \\ 
\end{matrix}\right|$.
$Z_3$, $Z_4$ and $Z_5$ correspond to 
$\displaystyle{\frac{1}{x}}$
$\displaystyle{\frac{1}{y_4}}$ and
$\displaystyle{\frac{1}{y_5}}$ respectively
and these relations correspond to (\ref{eq:rel345}).

\begin{gather*}
{\tiny{
\centerline{
\vbox{
	\baselineskip =10pt
	\tabskip = 1em
	\halign{&\hfil#\hfil \cr
        \multispan7 \hfil Table 1 \hfil \cr
	\noalign{\smallskip}
	\noalign{\hrule height0.8pt}
	\noalign{\smallskip}
$\ \ \ $ \strut\vrule & 
0 &1 & 2 & 3 & 4 & 5 & 6 & 7 & 8 & 9 & 10 & 11 & 12 & 13 & 8 & 9 & 10\cr
\noalign{\smallskip}
\noalign{\hrule height0.3pt}
\noalign{\smallskip}
$X_3$ \strut\vrule & 
 1& - & - & $x$ & $y_4$ & - & $x^2$& $x y_4$ & $y_4^2$ & $x^3$ 
& $x^2 y_4$ & $xy_4^2$& $x^4$ & $x^3y_4$ & $x^2y_4^2$ & $x^5$ & $x^4 y_4$\cr
$X_4$ \strut\vrule & 
 1& - & - & $x$ & - & $y_5$ & $x^2$& - & $xy_5$ & $x^3$ & $y_5^2$ & 
$x^2 y_5$ & $x^4$ & $x^2y_5$ & $x y_5^2$ & $x^5$ & $x^2y_5^2$ \cr
$X_2$ \strut\vrule & 
 1& - & - & $x$ & $y_4$ & $y_5$ & $x^2$& $x y_4$ & $x y_5$ & $y_4 y_5$ 
& $x^2 y_4$ 
& $x^2 y_5$ & $x^4$ & $x^3 y_4$ & $x^3y_5$ & $x^5$ & $x^4 y_4$ \cr
\noalign{\smallskip}
	\noalign{\hrule height0.8pt}
}
}
}
}}
\end{gather*}

There we define $\phi_{i}^{(g)}$ as a non-gap monomial 
in $R_g$ for $g= 2, 3, 4$ and 
{\it{e.g.}},
$\phi_{0}^{(2)} = 1$, 
$\phi_{1}^{(2)} = x$, 
$\phi_{2}^{(2)} = y_4$, 
$\phi_{3}^{(2)} = y_5$, 
$\phi_{4}^{(2)} = x^2$, 
$\cdots$ and
$\phi_{0}^{(3)} = 1$, 
$\phi_{1}^{(3)} = x$, 
$\phi_{2}^{(3)} = y_4$, 
$\phi_{3}^{(3)} = x^2$, 
$\phi_{4}^{(3)} = xy_4$, 
$\cdots$.
We introduce the weight $\Ng(n)$ by letting
$\Ng(n) := -\wt(\phi^{(g)}_{n})$,
where $\wt()$ is the degree of divisor at $\infty$ of each curve $X$'s.
It is noted that  $H_2$ is identical to $\{N^{(2)}(n)
 \ |\ n = 0, 1, 2, \ldots\}$.
For later convenience,
we also introduce $\phiHf{i} \in R$ $(i = 1, 2, 3, \cdots)$ by
$\phiHf{0} := y_4$, $\phiHf{1} := y_5$, 
$\phiHf{2} := x y_4$, $\phiHf{3} := x y_5$,
for $i > 3$,
$\displaystyle{
\phiHf{i} := \left\{ 
\begin{matrix}
x^{(i-4)/3} y_4 y_5 & i \equiv 1 \ \mathrm{mod} \ 3,\\
x^{(i+1)/3} y_4 & i \equiv 2 \ \mathrm{mod} \ 3,\\
x^{i/3} y_5 & i \equiv 0 \ \mathrm{mod} \ 3.\\
\end{matrix}
\right.
}$

We also define the weight $\NHf(n)$ by
$ \NHf(n) := -\wt(\phiHf{n})$;
$ \NHf(0) =4$,
$ \NHf(1) =5$,
and $ \NHf(n) =n + 5$ for $n\ge 2$.
By letting 
\begin{equation*}
\begin{split}
\Lambda_i^{(2)} &:= \NHf(2) - \NHf(i-1)  + i-3,\\
\Lambda_i^{(g)} &:= \Ng(g) - \Ng(i-1) -g + i -1, \quad (g=3,4)
\end{split}
\end{equation*}
the related Young diagrams, $\Lambda \equiv \Lambda^{(2)}
:=(\Lambda_1, \Lambda_2)=(1,1)
$,
$\Lambda^{(3)}
:=(\Lambda_1^{(3)}, \Lambda_2^{(3)},  \Lambda_3^{(3)})
=(3,1,1)$ 
and $\Lambda^{(4)}
:=(\Lambda_1^{(4)}, \Lambda_2^{(4)}, \Lambda_3^{(4)},
\Lambda_4^{(4)})
=(4,2,1,1)$  are given by respectively:
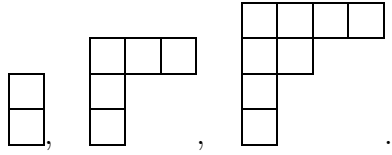
\begin{equation*}
\displaystyle{\yng(1,1)},\quad
\displaystyle{\yng(3,1,1)},\quad
\displaystyle{\yng(4,2,1,1)}.
\end{equation*}
The Young diagram $\Lambda$ is not symmetric, whereas
${}^t\Lambda^{(3)} = \Lambda^{(3)}$ and
${}^t\Lambda^{(4)} = \Lambda^{(4)}$.

Then the following propositions are obvious:

\begin{proposition}
Bases of the holomorphic one forms over $X$ are expressed by
$
\displaystyle{\nuIf{1} = \frac{ \dd x}{3y_5}}$ and $\displaystyle{\nuIf{2} = \frac{ \dd x}{3y_4}}$
or 
$\displaystyle{\nuIf{i} := \frac{\phiHf{i-1} \dd x}{3y_4 y_5}}$,
$(i=1,2)$. 
\end{proposition}

We note their divisors and linear equivlanece; for $B_a:=(b_a,0,0)$ $(a=0,1,2)$, 
 $(\nuIf{1})=\infty+B_0$  $\sim(\dd x/y_5^2) =2(3\infty-B_1-B_2)$
and  $(\nuIf{1})\sim(\nuIf{2})=B_1+B_2$ $\sim(\dd x/y_4^2)=2(2\infty-B_0)$
$=2(\infty+(\infty-B_0))$.

\begin{proposition}
$ \sum_{i = 0}^n a_i\tilde \nu_i $
belongs to $H^1(X\setminus \infty, \cO_{X})$,
where $\tilde \nu_i:= \frac{\phiHf{i} \dd x}{3 y_4 y_5}
$ and the order of the singularity of $(\tilde \nu_i)$ at $\infty$
is given by $\NHf(n) - 5$.
\end{proposition}

\begin{lemma}
$\displaystyle{
 a_0 \frac{\dd x}{ y_4 y_5}
+ a_1 \frac{x \dd x}{ y_4 y_5}
+ a_2 \frac{x^2 \dd x}{ y_4 y_5}}
$ is not holomorphic one form over $X$
if $a_i$ does not vanish.
\end{lemma}
\begin{proof}
For $n<3$, every 
$\sum_{i = 0}^n a_i \frac{x^i \dd x}{ y_4 y_5}$ has singularities
at points in $X \setminus \infty$. 
\end{proof}

We choose  the bases
$\alpha_{i}, \beta_{j}$  $ (1\leqq i, j\leqq 2)$
of $H_1(X,\ZZ)$ such that their intersection numbers are
$\alpha_{i}\cdot\alpha_{j}=\beta_{i}\cdot\beta_{j}= 0$ and 
$\alpha_{i}\cdot\beta_{j}=\delta_{ij}$, 
and we denote the period matrices by
$\displaystyle{ \left[\,\omegafp{}  \ \omegafpp{}  \right]= 
\frac{1}{2}\left[\int_{\alpha_{i}}\nuIf{j} \ \ \int_{\beta_{i}}\nuIf{j}
\right]_{i,j=1,2}}$.
Let $\Pi_2$ be a lattice 
generated by $\omegafp{}$ and $\ \omegafpp{}$.
For a point $P \in X$, the abelian map $\hu_o : X \to \CC^2$
is defined by
$$
\displaystyle{\hu_o(P) = \int^P_{\infty} \nuIf{} \in \CC^2}
$$
and for a point $(P_1,\cdots, P_k) \in S^k X$, i.e.,
the $k$-th symmetric product of $X$, the shifted abelian map
$\hu : S^k X \to \CC^2$ by
$$
\hu(P_1, \cdots, P_k) :=\hu_o(P_1, \cdots, P_k)+
\hu_o(B_0)
$$
where $\hu_o(P_1, \cdots, P_k) :=\sum_{i=1}^k \hu_o(P_i)$.
Then we define the Jacobian $\cJ_2$ and its subvariety $\WWf^k$ 
$(k = 0, 1, 2)$ by
$$
\kappa : \CC^2 \to \cJ_2 =\CC^2/\Pi_2 = \WWf^2,
   \quad \WWf^k := \kappa \hu(S^k X)
$$
respectively. Further the singular locus of 
$S^2 X$ is denoted by $S^2_1 X$ as in \cite{MP1}.

For a point $(P_1, P_2) \in S^2 X$ around the infinity point,
by letting their local parameters $t_{\infty,1}$ and $t_{\infty,2}$,
$u\equiv {}^t(u_1, u_2):=\hu_o(P_1, P_2)$ is given by
$u_1 = $ $\frac{1}{2}( t_{\infty,1}^2 +t_{\infty,2}^2)
(1+\dd_{>0}(t_{\infty,1},t_{\infty,2})),$
$ u_2 =$ $ ( t_{\infty,1} +t_{\infty,2}) 
(1+\dd_{>0}(t_{\infty,1},t_{\infty,2})),
$
where 
$\dd_{\ge}(t_{1},t_{2})$ is a natural extension of
$\dd_{\ge}(t)$.

\subsection{Differentials of the second and the third kinds}
\label{differential forms 2}
Following the EEL-construction \cite{EEL} for a $(n,s)$ curve,
we give an algebraic representation of
 a differential form which 
is equal to the fundamental normalized differential of the 
second kind in \cite[Corollary 2.6]{F},
up to a tensor of holomorphic one forms:
\begin{definition}
A two-form $\Omegaf(P_1, P_2)$ on $X\times X$ is called 
a {\it fundamental differential of the second kind} if it is symmetric, 
$\Omegaf(P_1, P_2)=\Omegaf(P_2, P_1)$, 
it has its only pole (of second order) along the diagonal of  $X\times X$, 
and in the vicinity of each point $(P_1,P_2)$ 
 is expanded in power series as 
\begin{equation}
\Omegaf(P_1, P_2)=\Big(\frac{1}{(t_{P_1} -t_{P_2} ')^2  } +\dd_\ge(1)\Big)
   \dd t_{P_1} \otimes \dd t_{P_2}
\ \ (\hbox{\rm as}\  P_1\rightarrow P_2)
\label{expansion}
\end{equation}
where $t_P$ is a local coordinate at a point $P \in X$.
\end{definition}

Here we use the convention that for $P_a \in X$,
$P_a$ is represented by $(x_a, y_{4,a}, y_{5,a})$
or $(x_{P_a}, y_{4,{P_a}}, y_{5,{P_a}})$
and for $P \in X$, $P$ is expressed by $(x, y_{4}, y_{5})$.
Then the following propositions holds.

\begin{proposition}
By letting 
$$\displaystyle{
   \Sigmaf\big(P, Q\big)
   :=\frac{ y_{4,P} y_{5,P} +y_{4,P} y_{5,Q} +y_{4,Q} y_{5,P}}
{(x_P - x_Q) 3 y_{4,P} y_{5,P} } \dd x_P
}$$
$\Sigmaf(P, Q)$ has the properties: 

\noindent
1) $\Sigmaf(P, Q)$ as a function of $P$ is singular at 
$Q=(x_Q, y_{4,Q}, y_{5,Q})$ and  $\infty$, and 
vanishes at $\hzeta_3^\ell(Q)= (x_Q, \zeta_3^\ell y_{4,Q}, 
\zeta_3^{2\ell} y_{5,Q})$, $(\ell = 1, 2)$, and

\noindent
 2)
$\Sigmaf(P, Q)$ as a function of $Q$ is singular at $P$ and
at $\infty$.
\end{proposition}

\begin{proof}
Direct computations lead the results.
\end{proof}

\begin{proposition} \label{prop:dSigma}
There exist  differentials $\nuIIf{j}=\nuIIf{j}(x,y_4,y_5)$ 
$(j=1, 2)$ 
of the second kind such that
they have
 their only pole at $\infty$ and satisfy the relation,
\begin{equation}
  \dd_{Q} \Sigmaf\big(P, Q\big) - 
  \dd_{P} \Sigmaf\big(Q, P\big)
  = \sum_{i = 1}^2 \Bigr(
         \nuIf{i}(Q)\otimes \nuIIf{i}(P)
        - \nuIf{i}(P)\otimes \nuIIf{i}(Q)
     \Bigr)
   \label{eq3.2}
\end{equation} 
where
$\displaystyle{
 \dd_{Q} \Sigmaf\big(P, Q\big)
   :=\dd x_P \otimes \dd x_Q\frac{\partial }{ \partial x_Q}
   \frac{
y_{4,P} y_{5,P}
+y_{4,P} y_{5,Q}
+y_{4,Q} y_{5,P}}
{(x_P - x_Q) 3 y_{4,P}  y_{5,P} } .
}$

The differentials $\{\nuIIf{1}, \nuIIf{2}\}$
are determined modulo the $\mathbb{C}$-linear space spanned by
$\langle\nuIf{j}\rangle_{j=1, 2}$;
we fix 
$$
\displaystyle{
\left\{\nuIIf{1}, \nuIIf{2}\right\}
}
\displaystyle{
 = 
\left\{ \frac{ -\left(2 x +\ltwo{1}\right) \dd x }{3 y_{4} },\ 
\frac{ -x \dd x }{3 y_{5} } \right\}
}
$$
as their representative.
\end{proposition}

\begin{proof}
$\displaystyle{\frac{\partial }{ \partial x_Q}
   \frac{ y_{4,P} y_{5,P} +y_{4,P} y_{5,Q} +y_{4,Q} y_{5,P}}
{(x_P - x_Q) 3 y_{4,P}  y_{5,P} } \dd x_P \\
}$ is equal to
{\small{
\begin{equation*}
\begin{split}
&\frac{1}{(x_P - x_Q) 
9y_{4,P} y_{5,P} y_{4,Q} y_{5,Q}} 
\Bigr[
\frac{
3(y_{4,P} y_{5,P} +y_{4,P} y_{5,Q} +y_{4,Q} y_{5,P})
y_{4,Q} y_{5,Q}}{(x_P - x_Q) }\\
& 
+ 
\Bigr(y_{4,P} 
\frac{y_{4,Q}}{y_{5,Q}}(
2k_{2,Q} k_{2,Q}' k_{1,Q} +k_{2, Q}^2  k_{1,Q}')
+ y_{5,P} 
\frac{y_{5,Q}}{y_{4,Q}}(
2k_{2,Q} k_{1,Q}k_{1,Q}'+ k_{2, Q}' k_{1,Q}^2))
\Bigr) \Bigr].\\
\end{split}
\end{equation*}
}}
Here $k_{a,P} = k_a(x_P)$ and $k_{a,P}' = \dd k_a(x_P)/\dd x_P$.
We have
$$
   \frac{\partial }{ \partial x_Q}
   \frac{ y_{4,P} y_{5,P} +y_{4,P} y_{5,Q} +y_{4,Q} y_{5,P}}
{(x_P - x_Q) 3 y_{4,P}  y_{5,P} }  
 -
   \frac{\partial }{ \partial x_P}
   \frac{ y_{4,Q} y_{5,Q} +y_{4,Q} y_{5,P} +y_{4,P} y_{5,Q}}
{(x_Q - x_P) 3 y_{4,Q}  y_{5,Q} }  
$$
\begin{equation*}
\begin{split}
&=\frac{1}{(x_P - x_Q) 
9y_{4,P} y_{5,P} y_{4,Q} y_{5,Q}} 
\left(B_2(P, Q) - B_2(Q, P)\right)
\end{split}
\end{equation*}
where
$B_2(P, Q) = 
y_{4,P} y_{5,Q}\Bigr(2x_Q +\ltwo{1}-x_P\Bigr)$.
Then we obtain the statements. 
\end{proof}


\begin{corollary}
\label{cor:Sigmaf}
1) The one form,
$
\Pif_{P_1}^{P_2}(P):= \Sigmaf(P, P_1) -  \Sigmaf(P, P_2),
$
is a differential of the third kind,  whose only 
(first-order) poles are
$P=P_1$ and $P=P_2$, and residues $+1$ and $-1$ 
respectively.  

2) $\displaystyle{
\Omegaf(P_1, P_2) 
}$ is defined by
$
\displaystyle{
 \dd_{P_2} \Sigmaf(P_1, P_2) 
     +\sum_{i = 1}^2 \nuIf{i}(P_1)\otimes \nuIIf{i}(P_2)
}$
$$
\Omegaf(P_1, P_2)   =\frac{\Ff(P_1, P_2)\dd x_1 \otimes \dd x_2}
{(x_{P_1} - x_{P_2})^2 9 
y_{4,P_1}
y_{5,P_1}
y_{4,P_2}
y_{5,P_2}}
$$
where $\Ff$ is an element of $R \otimes R$.
\end{corollary}
\begin{proof} 
Direct computations give the claims.
\end{proof} 

\begin{lemma} 
\label{lemma:limFphi4}
We have
$\displaystyle{
\lim_{P_1 \to \infty} 
\frac{\Ff(P_1, P_2)}{\phiHf{1}(P_1)(x_{P_1} - x_{P_2})^2}
 = \phiHf{2}(P_2) =x_{P_2} y_{4,P_2}}$.
\end{lemma}

\begin{proof}
$B_2$ in the proof of Proposition \ref{prop:dSigma}
leads the result.
\end{proof}

For later convenience we introduce the quantity,
$\displaystyle{\Omegaf^{P_1, P_2}_{Q_1, Q_2}
 := \int^{P_1}_{P_2} \int^{Q_1}_{Q_2} \Omegaf(P, Q)}$,
\begin{equation}
\label{eq:Omega_def2}
\Omegaf^{P_1, P_2}_{Q_1, Q_2}
 = \int^{P_1}_{P_2} (\Sigmaf(P, Q_1) - \Sigmaf(P, Q_2)) 
 +\sum_{i = 1}^4 \int^{P_1}_{P_2} \nuIf{i}(P) 
\int^{Q_1}_{Q_2} \nuIIf{i}(P).
\end{equation}

\section{The sigma function for $(3,4,5)$ curve}
\label{the sigma function}

\subsection{Generalized Legendre relation}
Corresponding to the complete integral of the first kind, we define the complete integral of the second kind,
$$
\displaystyle{
   \left[\,\etagp{}  \ \etagpp{}  \right]:= 
\frac{1}{2}\left[\int_{\alphag{i}}\nuIIg{j} \ \ 
                 \int_{\betag{i}}\nuIIg{j}
\right]_{i,j=1,2}}.
$$

Let $\taug{Q_1, Q_2}$ be the normalized 
differential of the third kind
such that $\taug{Q_1, Q_2}$ has residues $+1$ and $-1$ at $Q_1$ and $Q_2$
respectively,
is regular everywhere else,  and is  normalized,
 $\int_{\alpha_i} \taug{P, Q} = 0$ for $i = 1,2$ \cite[p.4]{F}.
The following Lemma corresponding 
 to  Corollary 2.6 (ii) in \cite{F} holds:
\begin{lemma} \label{lemma:2.1}
By letting $\gammag{} = {\omegagp{}}^{-1} \etagp{}$, we have
$$
{\Omegag}^{P_1, P_2}_{Q_1, Q_2} = 
\int^{P_1}_{P_2}
\taug{Q_1, Q_2} + \sum_{i, j = 1}^2 \gammag{ij} 
\int^{P_1}_{P_2} \nuIg{i} \int^{Q_1}_{Q_2} \nuIg{j}.
$$
\end{lemma}
\begin{proof}
The same as \cite[I: Lemma 4.1]{MP1}.
\end{proof}

The following Proposition provides a symplectic structure in the Jacobian
$\cJ_2$,
known as  {\it generalized Legendre relation}
\cite{BLE1,BLE2,MP1}:
\begin{proposition}
$\displaystyle{
   M\left[\begin{array}{cc} & -1 \\ 1 & \end{array}\right]{}^t {M}
   =2\pi\sqrt{-1}\left[\begin{array}{cc} & -1 \\ 1 &
     \end{array}\right]
}$
for 
$\displaystyle{
   M := \left[\begin{array}{cc}2\omegagp{} & 2\omegagpp{} \\
               2\etagp{} & 2\etagpp{}
     \end{array}\right].
}$
\end{proposition}
\begin{proof}
The same as \cite[I: Propositon 4.2]{MP1}.
\end{proof}

\subsection{The $\sigma$ function}
Due to the  Riemann relations \cite{F}, 
 $\text{Im}\,({\omegagp{}}^{-1}\omegagpp{}) $ is positive definite.
Theorem 1.1 in \cite{F} gives
$\displaystyle{
   \delta:=\left[\begin{array}{cc}\delta''\ \\
       \delta'\end{array}\right]\in \left(\frac{\ZZ}{2}\right)^{4}
}$
be the theta characteristic which is equal to the Riemann constant 
$\xi_R$ and the period matrix 
$[\,2\omegagp{}\ 2\omegagpp{}]$. We note that $\xi_R=$
$\hu(P_R)$ for a point $P_R\in X$ satisfying $2P_R+2B_0-4\infty\sim0$.
 We define an entire function of (a column-vector)
$u={}^t\negthinspace (u_1, u_2)\in \mathbb{C}^2$,
$$
   \sigma{}(u)
   =c\mathrm{e}^{-\tfrac{1}{2}\ ^t\negthinspace 
u\etagp{}{\omegagp{}}^{-1}\  u} 
   \sum_{n \in \ZZ^2} \mathrm{e}^{\big[\pi \sqrt{-1}\big\{
    \ ^t\negthinspace (n+\delta'')
      {\omegagp{}}^{-1}\omegagpp{}(n+\delta'')
   + \ ^t\negthinspace (n+\delta'')
      ({\omegagp{}}^{-1} u+\delta')\big\}\big]}
$$
where  $c$ is a certain constant as in (\ref{eq:Schur}).

For a given $u\in\CC^2$, we  introduce
$u'$ and $u''$ in $\RR^2$ so that
$u=2\omegagp{} u'+2\omegagpp{} u''$. 

\begin{proposition} \label{prop:pperiod}
For $u$, $v\in\CC^2$, and $\ell$
$(=2\omegagp{}\ell'+2\omegagpp{}\ell'')$ $\in\Pi_2$, by letting
$ L(u,v) $  $ :=2\ {}^t{u}(\eta'v'+\eta''v'')$,
$\chi(\ell):=\exp[\pi\sqrt{-1}\big(2({}^t {\ell''}\delta'-{}^t
  {\ell'}\delta'') +{}^t {\ell'}\ell''\big)]$, 
we have a translational relation, 
$$
\sigma{}(u + \ell) = 
\sigma{}(u) \exp(L(u+\frac{1}{2}\ell, \ell)) \chi(\ell).
$$
\end{proposition}
\begin{proof}
The same as \cite[I: Prop.4.3]{MP1}.
\end{proof}

The vanishing locus of $\sigma{}$ is simply given by
$\Theta^{1} :=( \WW^{1} \cup [-1] \WW^{1}) = \WW^{1}$.

\subsection{The Riemann fundamental relation}
As in \cite[I: Prop4.4]{MP1},
we have the Riemann fundamental
relation:
\begin{proposition}\label{prop:RR}
For $(P, Q, P_i, P'_i) \in X^2 \times (S^2(X)\setminus S^2_1(X)) \times 
(S^2(X)\setminus S^2_1(X))$, 
\begin{align*}
\exp\left( 
\sum_{i, j = 1}^2 
   {\Omegag}_{P_i, P'_j}^{P, Q} \right)
&=\frac{\sigma{}(\uab_o(P) - \uab(P_1,  P_2)) 
        \sigma{}(\uab_o(Q) - \uab(P'_1, P'_2))}
     {\sigma{}((\uab_o(Q) - \uab(P_1, P_2)) 
      \sigma{}(\uab_o(P) - \uab(P'_1, P'_2))}.
\end{align*}
\end{proposition}

Using the differential identity,
$\displaystyle{
          \sum_{i, j=1}^2 \phiHf{i-1}(P'_1) \phiHf{j-1}(P'_2)
	\frac{\partial^2 }
{\partial \hu_i(P'_1)\partial \hu_j(P'_2)}=
}$\\
$\displaystyle{
 9 y_{4,P'_1} y_{5,P'_1} y_{4,P'_2} y_{5,P'_2}
\frac{\partial^2 }{\partial x'_1\partial x'_2},
}$
taking logarithm of both sides of the relation
and differentiating them
along $P'_1=P$ and $P'_2=P_a$, 
we have the differential expressions of the relation,
as mentioned in \cite[I: Prop. 4.5]{MP1}:
\begin{proposition} \label{prop:wpxx=F}
For $(P, P_1, P_2) \in X \times S^2(X) \setminus S^2_1(X)$
and $u := \uab(P_1, P_2)$,
the equality
$$
	\sum_{i, j = 1}^2 \wpg{i, j}
          \left( \uab_o(P) - u\right)
         \phiHf{i-1}(P)
         \phiHf{j-1}(P_a) =\frac{\Fg(P, P_a)}{(x-x_a)^2}
$$
holds for every $a = 1, 2$,
where we set
$$
	\wpg{ij}(u) := 
-\frac{\sigma_{i}(u) \sigma_{j}(u) 
              -  \sigma{}(u) \sigma_{ij}(u)}
               {\sigma{}(u)^2}\equiv
        -\frac{\partial^2}{\partial u_i\partial u_j}\log\sigma{}(u).
$$
\end{proposition}

\subsection{Jacobi inversion formulae}
As in \cite{MP1},
we introduce meromorphic functions on the curve $X$:
\begin{definition} \label{def:mul}
For $P, P_1, \ldots, P_n$ $\in (X\backslash\infty) \times \SS^n(X\backslash\infty)$, 
$(n=1,2)$,
we define 
\begin{gather*}
\begin{split}
\mug{1}(P; P_1)&:= 
y_5 -\frac{y_{5,1}}{y_{4,1}}y_4,\\
\mug{2}(P; P_1, P_2)&:= x y_4
-\frac{y_{4,1} x_2y_{4,2} - y_{4,2} x_1y_{4,1}}
{y_{4,1} y_{4,2} - y_{4,2} y_{4,1}} y_5
+\frac{y_{5,1} x_2y_{4,2} - y_{5,2} x_1y_{4,1}}
{y_{4,1} y_{4,2} - y_{4,2} y_{4,1}}y_4.
\end{split}
\end{gather*}
\end{definition}

We note that  $\mug{n}$ for $X$ is characterized by
 the condition on a polynomial $\mug{n}
= \sum_{i=0}^{n} a_i \phiHf{i}(P)$, $a_i \in \CC$ and $a_n = 1$,
 which has a  zero at each point $P_i$
and has the smallest possible order such 
that it multiplied by $d x/3 y_4 y_5$ belongs to 
$H^1(X \setminus \infty, \cO_X)$.
For given $P_1$, the solution of $\mug{1}(P; P_1)=0$ corresponds to a
point $Q_1=[-1]P_1$ with $B_a$ $(a=0,1,2)$, and 
for given $P_1$ and $P_2$, the solution of $\mug{2}(P; P_1,P_2)=0$ gives
two points $Q_1,Q_2$ with $B_a$ $(a=0,1,2)$ such that
$Q_1+Q_2=[-1](P_1+P_2)$. Here we use $B_0+B_1+B_2-3\infty\sim2B_0-2\infty$.

Using $\mug{n}$,
we have our main theorem in this article:
\begin{theorem} \label{prop:4.5}
1) For $(P, P_1, P_2) \in X \times \left(S^2(X) \setminus S^2_1(X)\right)$,
we have
\begin{enumerate}
\item[1-1)] 
$
    \mug{2}(P; P_1, P_2) = 
    x y_4- \wp_{2 2}(\uab(P_1,P_2)) y_4 +\wp_{2 1}(\uab(P_1,P_2)) y_5.
$

\item[1-2)] 
$ \displaystyle{\wp_{2 2}(\uab(P_1,P_2)) = 
\frac{y_{4,1} x_2y_{4,2} - y_{4,2} x_1y_{4,1}}
{y_{4,1} y_{4,2} - y_{4,2} y_{4,1}}}$

$\displaystyle{\wp_{2 1}(\uab(P_1,P_2)) = 
\frac{y_{5,1} x_2y_{4,2} - y_{5,2} x_1y_{4,1}}
{y_{4,1} y_{4,2} - y_{4,2} y_{4,1}}}\cdot$
\end{enumerate}

\noindent
2) For $(P, P_1) \in X\times (X \setminus S^{1}_1(X))$ and
$u = \uab(P_1)\in \kappa^{-1}(\WW^{1})$,
$$
\mug{1}(P; P_1)= 
y_5 -\frac{\sigma_{1}(u)}
               {\sigma_{2}(u)}y_4,\quad\mbox{and}\quad
\frac{\sigma_{1}(u)}
               {\sigma_{2}(u)}= \frac{y_5}{y_4}\cdot
$$
\end{theorem}

\begin{proof}
1) is the same as \cite[I: Prop. 4.6]{MP1}.
As in \cite[I: Theorem 5.1]{MP1}, by considering
$\displaystyle{
\lim_{P_2 \to \infty} \frac{\wp_{21}(\hu(P_1,P_2))}{\wp_{22}(\hu(P_1,P_2))}
}$, we have the second result.
\end{proof}

Following the statement by
Buchstaber,  Leykin and  Enolskii,
Nakayashiki showed that
the leading of the sigma function for $(r,s)$ curve 
is expressed by Schur function \cite{MP2}.
Noting (\ref{eq:x_t}) and degrees of $u$, the above
Jacobi inversion formulae gives an extension that
$$
\displaystyle{
   \sigma{}(u) = \frac{1}{2}u_2^2 - u_1 
            + \sum_{|\alpha|>2} a_\alpha u^\alpha
}
$$
where $a_\alpha\in \QQ[b_1, \cdots, b_5]$,
$\alpha = (\alpha_1, \alpha_2)$,
$|\alpha| = \alpha_1 + \alpha_2$
and  $u^\alpha = u_1^{\alpha_1} u_2^{\alpha_2}$.
The prefactor $c$ is determined by this relation.
Since for a Young diagram $\Lambda$,
$S_\Lambda$ and $s_\Lambda$ are the Schur functions defined by
\begin{equation}
S_{\Lambda}(T_1, T_2) = t_1 t_2 = \frac{1}{2} T_1^2 - T_2
\label{eq:Schur}
\end{equation}
where $T_1 :=t_1 + t_2$ and $T_2 :=\frac{1}{2}(t_1^2 + t_2^2)$,
we have
$$
\displaystyle{
   \sigma{}(u) = S_{\Lambda}(u_1, u_2)
}
\displaystyle{
            + \sum_{|\alpha|>2} a_\alpha u^\alpha.
}
$$

\begin{remark}
\label{rmkF}
{\rm{
We showed that the EEL construction works well even for 
a space curve, and the sigma function associated with the curve
is naturally defined.
Since this construction is very natural,
 this study sheds a new light on the way to 
construction of the sigma functions for space curves.
We conjectured that the EEL construction
 could be applied to  every space curve if it is Weierstrass.

\smallskip

As an interesting example of a space curve, 
we will give a comment on a problem as follows, 
for which we started to study sigma functions for affine space curves.

McKay considers a relation between dispersionless
KP hierarchy and the replicable functions in order to obtain
a further profound interpretation of the moonshine phenomena
of Monster group \cite{MS}.
He conjectured that it might be related to the quantised elastica \cite{Ma08,Ma10}.
By studying a relation between a replicable function and
an algebraic curve associated with elastica,
Matsutani found that a semigroup $H_{12}$ generated by 
$M_{12}:=\LA6,13,14,15,16\RA$
has gap sequence,
$
L(H_{12})=\{1,2,3,4,5,7,8,9,10,11,17,23\},
$
which is identical to the Norton number,
$
N_{12}:=\{1,2,3,4,5,7,8,9,11,17,19,23\}
$
by exchanging
10 and 19. 
The Norton number plays the essential role in the moonshine phenomena
for the Monster group \cite{MS}. 
 The replicable function is given as an element of
$\QQ[a_1, a_2, a_3,a_4,$ $a_5,a_7,a_8,a_9,a_{11},$ $a_{17},a_{19},a_{23}][[t]]$.
The replicable function is a generalization of the elliptic $J$-function,
which causes the moonshine phenomena of the Monster group.

After then,
Komeda proved that $H_{12}$ is the Weierstrass semigroup
and gave the fundamental relations 
Propositions \ref{prop:Z12} 
 as mentioned in Appendix,  which is
 reported more precisely in \cite{KMP}.
Then we applied the EEL-construction to the curve and
obtain a sigma function for 
a Jacobi variety $\cJ_{12}$ for $C(H_{12})$ \cite{KMP}.
Since the Jacobi variety $\cJ_{12}$ is given as 12-dimensional complex
torus whose real dimension is 24,
it might remind us of Witten conjecture associated with 
Monster group problem \cite{HBJ};
Witten conjectured that a 24 dimensional manifold exists
such that the Monster group acts on it via Weierstrass sigma function.
}}
\end{remark}

\setcounter{section}{0}
\renewcommand{\thesection}{\Alph{section}}
\section{Appendix: Weierstrass properties of $(6,13,14,15,16)$
by Jiryo\\
 Komeda}

\smallskip
The proofs of these propositions are given in the article \cite{KMP}
 in detail.  We show only the sketch of the first one
because the second one is not difficult.

\begin{proposition} \label{prop:H12}
The numerical semigroup $\LA6, 13, 14, 15, 16\RA$ is Weierstrass.
\end{proposition}

\begin{proof}
Let $(C, P)$ be a pointed curve with $H(P) = \LA3, 7, 8\RA$. 
Then 
$$
2 = h^0(4P) = 4 + 1 - 4 + h^0(K - 4P) = 1 + h^0(K - 4P)
$$
which implies that $K - 4P \sim P_1 + P_2$ for some points 
$P_1$ and $P_2 \in C$. Here $K$ is a canonical divisor on $C$. 
Moreover,
$$
2 = h^0(5P) = 5 + 1 - 4 + h^0(K - 5P) = 2 + h^0(K - 5P)
$$
which implies that $h^0(K - 5P) = 0$. 
Hence, we get $P_i \neq P$ for $i = 1, 2$. 
Thus, $K \sim 4P + P_1 + P_2$ with $P_i \neq P$ for $i = 1, 2$. 
We set $D = 7P - P_1 - P_2$. Then deg$(2D - P) = 9 = 2 \times 4 + 1$,
which implies that the complete linear system $|2D - P|$ is 
very ample, hence base-point free. Therefore,
$2D \sim P + Q_1 + \ldots + Q_9$ ($=$ a reduced divisor).
Let $\cL$ be the invertible sheaf $\cO_C(-D)$ on $C$ and $\phi$
an isomorphism $\cL^{\otimes2} \approx \cO_C(-P-Q_1-\cdots-Q_9) \subset
\cO_C$. 
Then the vector bundle $\cO_C \oplus \cL$
 has an $\cO_C$-algebra structure through $\phi$. 
The canonical morphism $\pi : \tilde C 
= \Spec(\cO_C \oplus \cL) \to C$, 
is a double covering. Its branch locus of $\pi$
 is $\{P,Q_1, . . . ,Q_9\}$. 
Let $\tilde P$ be the ramification point of $\pi$ over $P$. 
Then it can be showed that $H(\tP) = \LA6, 13, 14, 15, 16\RA$
 using the formula,
$ h^0(2n \tP) = h^0(nP) + h^0(nP - D)$
for any non-negative integer $n$. 

By considering $h^0(2n \tP)$ for $n=3,4,5,6,7,8,9$,
we show $H(\tP) = \LA6, 13, 14,$ $ 15, 16\RA$.
\end{proof}

\begin{proposition} \label{prop:Z12} 
Let $B_{12}$ a monomial ring which is given by
$k[t^a]_{a \in M_{12}}$
for the numerical semigroup $H_{12}$.
For a  $k$-algebra homomorphism,
$$
	\varphi_{12} : k[Z] := k[Z_6, Z_{13}, Z_{14}, Z_{15}, Z_{16}] 
\to k[t^a]_{a \in M_{12}}
$$
where $Z_a$ is the weight of $a = 6, 13, 14, 15, 16$,
the kernel of $\varphi_{12}$ is generated by the following
relations $\fZ{12,b}$ $(b = 1, \cdots, 9)$,
\begin{equation*}
\begin{array}{lll}
\fZ{12,1} = Z_{13}^2 - Z_{6}^2 Z_{14}, \quad
&\fZ{12,2} = Z_{13} Z_{14} - Z_{6}^2 Z_{15}, \quad
&\fZ{12,3} = Z_{14}^2 - Z_{13} Z_{15} , \\
\fZ{12,4} = Z_{14}^2 - Z_6^2 Z_{16}, \quad
&\fZ{12,5} = Z_{13} Z_{16} - Z_{14} Z_{15} , \quad
&\fZ{12,6} = Z_{15}^2 - Z_6^5, \\
\fZ{12,7} = Z_{14} Z_{16} - Z_6^5, \quad
&\fZ{12,8} = Z_{15} Z_{16} - Z_6^3 Z_{13} ,\quad
&\fZ{12,9} = Z_{16}^2 - Z_6^3 Z_{14}.\\
\label{eq:Z12}
\end{array}
\end{equation*}
\end{proposition}


\section*{Acknowledgements}
 One of the authors (S.M.) thanks John McKay for posing my attention to the Norton problem
and his encouragement.
This work started in a seminar at Yokohama National university 2008,
and was stimulated by
the international conference at HWK 2011.
S.M. is grateful to Kenichi Tamano, Norio Konno, Claus L\"ammerzahl,
Jutta Kunz, and Victor Enolskii.
S.M. is also most grateful to Emma Previato, Yoshihiro \~Onishi and Yuji Kodama for 
crucial discussions.

\bigskip

\smallskip
\noindent
Shigeki Matsutani:\\
8-21-1 Higashi-Linkan Minami-ku,\\
Sagamihara 252-0311,\\
JAPAN.\\
e-mail: {rxb01142@nifty.com}\\

\noindent
Jiryo Komeda\\
Department of Mathematics \\
Center for Basic Education and Integrated Learning\\
Kanagawa Institute of Technology\\
Atsugi, 243-0292 \\
JAPAN.\\
e-mail: {komeda@gen.kanagawa-it.ac.jp}\\

\end{document}